\documentclass[review,times,authoryear,english,a4paper,12pt]{elsarticle}
\usepackage[margin=1in]{geometry} 
\usepackage[OT1]{fontenc}
\usepackage{babel}
\usepackage{natbib}

\usepackage[utf8]{inputenc}

\usepackage[colorlinks]{hyperref}

\usepackage{amsmath,amssymb,amsthm}

\usepackage{mathrsfs}

\usepackage{graphicx}

\usepackage{color}

\usepackage{indentfirst}

\usepackage{multirow}

\usepackage{float}

\usepackage{subfig}

\usepackage{enumerate} 

\allowdisplaybreaks

\usepackage{hyperref}

\usepackage{setspace}

\newtheorem{thm}{Theorem}
\newtheorem{lem}{Lemma}
\newtheorem{remark}{Remark}
\journal{Journal of Financial Econometrics}

\begin{document}

	\begin{frontmatter}
		\title{Forecasting High-Dimensional Realized Volatility Matrices Using A Factor Model}
		\author{Keren Shen \tnoteref{mytitlenote}}
		\address{Department of Statistics and Actuarial Science\\		
			The University of Hong Kong\\
			Pokfulam, Hong Kong}
		\tnotetext[mytitlenote]{Corresponding Author:}
		\ead{rayshenkr@hku.hk}

		\author{Jianfeng Yao}
		\address{Department of Statistics and Actuarial Science\\		
			The University of Hong Kong\\
			Pokfulam, Hong Kong}
		\author{Wai Keung Li}
		\address{Department of Statistics and Actuarial Science\\		
			The University of Hong Kong\\
			Pokfulam, Hong Kong}
		
		\begin{abstract}
			Modeling and forecasting covariance matrices of asset returns play a crucial role in finance. The availability of high frequency intraday data enables the modeling of the realized covariance matrix directly. However, most models in the literature suffer from the curse of dimensionality. To solve the problem, we propose a factor model with a diagonal CAW model for the factor realized covariance matrices. Asymptotic theory is derived for the estimated parameters. In an extensive empirical analysis, we find that the number of parameters can be reduced significantly. Furthermore, the proposed model maintains a comparable performance with a benchmark vector autoregressive model.
		\end{abstract}
		
		\begin{keyword}
			High-dimension \sep High-frequency \sep Realized covariance matrices \sep Factor model \sep Wishart distribution
			\JEL C31\sep C38 \sep C55
		\end{keyword}
		
	\end{frontmatter}
 \doublespacing
 \newpage
 \section{Introduction}
 \noindent
 Modeling and forecasting covariances or volatility matrices of asset returns play a crucial role in many financial fields, such as portfolio allocation \citep{M52} and asset pricing \citep{BEW88}. With the availability of intraday financial data nowadays, it becomes possible to estimate volatilities and co-volatilities of asset returns using high-frequency data directly, which leads to the so-called {\em realized covariance matrix} (\citealp{ABDL03}, \citealp{BS04} and \citealp{BHLS11}). Two major problems arise in the estimation of realized covariance matrices. Firstly,  
 transactions for different assets are typically asynchronous so that the high-frequency prices of different assets do not change simultaneously. Secondly, it is widely believed that the observed high-frequency prices are accompanied
 by {\em microstructure noise} so that the observed prices should be thought as a noisy version of the true underlying price process. Researchers have proposed several ways to tackle these problems, for example, the overlap intervals and previous tick method by \citet{HY05} and \citet{Z11}, respectively. Moreover, \citet{BMOV12} use a refresh time scheme and \citet{CKP10} propose the pre-averaging approach.
 
 Once constructed, realized covariance matrices are analyzed using multivariate models. There are two major issues to be resolved in modeling realized covariance matrices. The first issue is that the model should guarantee the positive definiteness of fitted covariance matrices. A natural choice in this aspect is the family of matrix-valued Wishart distributions which automatically generates random positive definite matrices without imposing additional constraints. Several models
 related to the Wishart distribution have been put forward. \citet{GJS09} propose the Wishart Autoregressive (WAR) model where the realized covariance matrix has a conditional distribution which is noncentral Wishart with a non-centrality parameter depending on lagged covariances and a fixed scaling matrix. Later, \citet{GGL12} propose the Conditional Autoregressive Wishart (CAW) model under which the conditional distribution is central Wishart with time dependent scaling matrices. Moreover, \citet{NLY14} generalize the above two models to construct Generalized Conditional Autoregressive Wishart (GCAW) model. There are other models involving the Wishart distribution, for instance, see \citet{JM09}.
 
 The second issue in modeling realized covariance matrices is the high-dimensionality. Indeed, covariance matrices have $d(d+1)/2$ entries for $d$ assets; consequently, the number of parameters in the model for realized covariance matrices grows quickly with $d$. For example, for an unrestricted CAW(2,2) model with $d=10$ assets, as many as 456 parameters are needed so that
 it is quite challenging to fit such a model in practice. This is probably the major reason why all the empirical studies we find in the literature on model fitting for realized covariance matrices are all limited to a {\em small number}, say 3 to 5 assets. Another problem inherent in realized covariance matrices is that they deviate from their population counterpart, the so-called {\em integrated covariance matrix}, when the number of assets is large compared to the sample size (\citealp{JL09}; and \citealp{WZ10}). It is therefore important to build statistical models for realized covariance matrices with a large dimension, say several
 tens.

 Improved estimators of realized covariance matrices are proposed in \citet{WZ10} with a so-called averaging realized volatility matrix (ARVM) estimator. \citet{TWYZ11} propose the threshold averaging realized volatility matrix (TARVM) estimator which is of two-scale and uses the previous-tick method and the threshold technique in constructing realized volatility matrices. Then, inspired by \citet{Z06} and \citet{FW07}, \citet{TWC13} propose the threshold multi-scale realized volatility matrix (TMSRVM) estimator. The TARVM and TMSRVM estimators are shown to be consistent for the integrated covariance matrix when the dimension of the realized covariance matrix, the sample size of intraday points and the length of sampling days go to infinity. In addition, the TMSRVM estimator proves to have the optimal convergence rate under the existence of the microstructure noise.
 
 As an effort to control the parametric dimension of the models, \citet{TWYZ11} propose to first identify a small number of factors for the realized covariance matrices and to fit a vector autoregressive (VAR) model to the vectorized factor covariance matrices. They show that such a factor model significantly reduces the number of parameters needed to fit the realized covariance matrices. However, the VAR specification is not able to ensure the positive definiteness of the predicted factor covariance matrices. In addition (see below), such a VAR fit still needs $\mathcal{O}(r^4)$ number of parameters with $r$ factors.
 
 In this article, we adopt the factor model approach introduced in \citet{TWYZ11} for realized covariance matrices. In order to overcome the aforementioned weakness
 of their VAR fit for the extracted factors, we propose a diagonal CAW model which has several advantages. Firstly, the proposed CAW model is able to guarantee automatically the positive definiteness of the covariance matrices generated from the model without imposing additional constraints. Secondly, as will be shown by extensive data analysis reported in this paper, 
 our model has excellent empirical performance in terms of the reduction of number of parameters compared to the VAR approach proposed in \citet{TWYZ11}: indeed we obtain comparable forecasting performance with much less parameters. 
 
 In a related work, \citet{AM14} also use a combination of factor extraction and CAW modeling and report some empirical studies with $7$ assets which is still considered to be small dimension. In this paper, we focus on a larger collection of assets where empirical studies are carried out for 30 assets.  
 A further difference in this paper is that we also propose a thorough theoretical analysis of both the factor modeling and the CAW estimation. 
 
 The rest of the paper is organized as follows. Section 2 introduces the model setup and our approach based on a factor model and a diagonal CAW model for the extracted factors. In Section 3, the asymptotic theory is established. In Sections 4 and 5, we report the middle and large scale data analysis on asset prices, respectively. Conclusions are presented in Section 6. Proofs of the asymptotic theory are provided in the Appendix.\\
 
 \section{Methodology}
 \subsection{Model setup}
 \noindent
 Suppose there are $d$ assets and their log price process $\mathbf{X}(t) = \{X_1(t),\cdots,X_d(t)\}'$ follows a continuous diffusion model:
 \begin{equation}
 	d\mathbf{X}(t) = \mathbf{\mu}_tdt + \mathbf{\sigma}_td\mathbf{W}_t, \quad t \in [0,T],
 \end{equation}
 where $\mathbf{\mu}_t$ is a drift in $\mathbb{R}^d$, $\mathbf{W}_t$ a standard $d$-dimensional Brownian motion and $\mathbf{\sigma}_t$ a $d \times d$ matrix. The {\em integrated volatility matrix} for the $t$-th day is defined as
 \begin{equation}
 	\mathbf{\Sigma}_x(t) = \int_{t-1}^{t} \mathbf{\sigma}_s\mathbf{\sigma}_s' ds, \quad t = 1,\cdots,T.
 \end{equation}
 However, it is commonly admitted that the microstructure noise is inherent in the high-frequency price process so that we are not able to observe directly $X_i(t)$, but $Y_i(t_{i\ell})$, a noisy version of $X_i(\cdot)$ at times $t_{i\ell}$, $\ell=1,\cdots, n_i$, $i = 1, \cdots, d$. Here, $n_i$ is the total trading times and $t_{i\ell}$ is the $\ell$-th trading times of asset $i$ during a giving trading day $t$. The observations $Y_i(t_{i\ell})$ is allowed to be non-synchronized, i.e. $t_{i\ell} \ne t_{j\ell}$ for any $i \ne j$. In this paper, we assume that
 \begin{equation}
 	Y_i(t_{i\ell}) = X_i(t_{i\ell}) + \epsilon_i(t_{i\ell})
 \end{equation}
 where $\epsilon_i(t_{i\ell})$ are i.i.d. microstructure noise with mean zero and variance $\eta_i$, and $\epsilon_i(\cdot)$ and $X_i(\cdot)$ are independent with each other.\\
 
 \subsection{Realized covariance matrix estimator}
 \noindent
 Several issues arise for the estimation of $\mathbf{\Sigma}_x(t)$: 1. asynchronous observations of different assets; 2. microstructure noise; 3. the number of assets can be larger than the sample size. In this paper, we adopt the threshold Multi-Scale Realized Volatility Matrix estimator (threshold MSRVM) proposed by \citet{TWC13}, denoted by $\mathbf{\hat{\Sigma}}_x(t)$, $t = 1,\cdots,T$. The threshold MSRVM estimator has many attractive properties, for instance, it is consistent for the high-dimensional integrated co-volatility matrix with the optimal convergence rate. Briefly, the idea of the threshold MSRVM estimator is the following: the previous-tick method is used to construct the raw realized covariance matrices. Then, a multi-scale estimator is evaluated which is actually a kind of average of those raw estimators. In addition, the multi-scale estimator is regularized using a thresholding method, that is, matrix entries under a threshold are set to be zero.\\
 
 \subsection{Matrix factor model}
 \noindent
 We adopt the factor model proposed by \citet{TWYZ11} to reduce the large dimension of $\mathbf{\Sigma}_x(t)$:
 \begin{equation}
 	\mathbf{\Sigma}_x(t) = \mathbf{A}\mathbf{\Sigma}_f(t)\mathbf{A}'+\mathbf{\Sigma}_0,
 \end{equation}
 for $t=1,\cdots,T$, where $\mathbf{\Sigma}_f(t)$ are $r \times r$ positive definite factor covariance matrices, $\mathbf{\Sigma}_0$ is a $d \times d$ positive definite constant matrix and $\mathbf{A}$ is a $d \times r$ factor loading matrix normalized by the constraint $\mathbf{A}'\mathbf{A}=\mathbf{I}_r$. As in a standard factor model, only the left-hand side of Equation (4) is observed. The unknown quantities are estimated using the method in \citet{TWYZ11}. Let
 \begin{equation}
 	\mathbf{\bar{\Sigma}}_x=\frac{1}{T}\sum_{t=1}^{T}\mathbf{\Sigma}_x(t), \quad \mathbf{\bar{S}}_x=\frac{1}{T}\sum_{t=1}^{T}\{\mathbf{\Sigma}_x(t)-\mathbf{\bar{\Sigma}}_x\}^2,
 \end{equation}
 and
 \begin{equation}
 	\mathbf{\bar{\hat{\Sigma}}}_x=\frac{1}{T}\sum_{t=1}^{T}\mathbf{\hat{\Sigma}}_x(t), \quad \mathbf{\bar{\hat{S}}}_x=\frac{1}{T}\sum_{t=1}^{T}\{\mathbf{\hat{\Sigma}}_x(t)-\mathbf{\bar{\hat{\Sigma}}}_x\}^2.
 \end{equation}
 Next, the estimator $\mathbf{\hat{A}}$ is obtained using the $r$ orthonormal eigenvectors of $\mathbf{\bar{\hat{S}}}_x$, corresponding to its $r$ largest eigenvalues, as its columns. Finally, the estimated factor covariance matrices are
 \begin{equation}
 	\mathbf{\hat{\Sigma}}_f(t) = \mathbf{\hat{A}}'\mathbf{\hat{\Sigma}}_x(t)\mathbf{\hat{A}}
 \end{equation}
 for $t = 1,\cdots,T$; and $\mathbf{\Sigma}_0$ is estimated by
 \begin{equation}
 	\mathbf{\hat{\Sigma}}_0 = \mathbf{\bar{\hat{\Sigma}}}_x - \mathbf{\hat{A}}\mathbf{\hat{A}}'\mathbf{\bar{\hat{\Sigma}}}_x\mathbf{\hat{A}}\mathbf{\hat{A}}'.
 \end{equation}

 \subsection{CAW modeling for factor covariance matrix}
 \noindent
 With the estimated factor covariance matrices $\{\mathbf{\hat{\Sigma}}_f(t)\}$ calculated by (7) and (8), we construct a dynamic structure by fitting a diagonal CAW model to $\mathbf{\tilde{\Sigma}}_f(t)$, where $\mathbf{\tilde{\Sigma}}_f(t) := \mathbf{\Sigma}_f(t)+\mathbf{A}'\mathbf{\Sigma}_0\mathbf{A}$. Here, $\mathbf{\tilde{\Sigma}}_f(t)$ is modeled rather than $\mathbf{\Sigma}_f(t)$ since $\mathbf{\hat{\Sigma}}_f(t)$ is in fact a consistent estimator of the former, and it is impossible to construct a consistent estimator for the latter, to our knowledge.
 
 The model is defined as follows. Let $\mathcal{F}_{t-1} = \sigma(\mathbf{\tilde{\Sigma}}_f(s),s<t)$ be the past history of the process at time $t$. Conditional on $\mathcal{F}_{t-1}$, $\mathbf{\tilde{\Sigma}}_f(t)$ follows a central Wishart distribution
 \begin{equation}
 	\mathbf{\tilde{\Sigma}}_f(t)|\mathcal{F}_{t-1} \sim \mathcal{W}_n(\nu,\mathbf{S}_f(t)/\nu),
 \end{equation}
 with $\nu$ the degrees of freedom and the scaling matrix $\mathbf{S}_f(t)$. Moreover, the scaling matrix $\mathbf{S}_f(t)$ follows a linear recursion of order $(p,q)$
 \begin{equation}
 	\mathbf{S}_f(t) = CC' + \sum_{i=1}^{p}B_i\mathbf{S}_f(t-i)B_i'+\sum_{j=1}^{q}A_j\mathbf{\tilde{\Sigma}}_f(t-j)A_j',
 \end{equation}
 where $A_j$, $B_i$ and $C$ are all $r \times r$ matrices of coefficients.
 
 In summary, the CAW process depends on the parameters $\{\nu,C, (B_i)_{1 \le i \le p}, (A_j)_{1 \le j \le q}\}$ without additional constraints, so that the total number of parameters is equal to $(p+q)r^2+\frac{r(r+1)}{2}+1= \mathcal{O}(r^2)$ which still grows quickly with the number of factors $r$ and the order $p$ and $q$. Since the main aim of the paper is to propose a practically feasible model for a large number of assets while retaining efficiency, we will restrict ourselves to {\em diagonal} coefficient matrices $C$, $(B_i)_{1 \le i \le p}$ and $(A_j)_{1 \le j \le q}$. Therefore, the number of parameters becomes $(p+q+1)r+1 = \mathcal{O}(r)$. Notice that this setup is also supported by the fact that in the literature (\citealp{MS92}, \citealp{EK95}), researchers tend to use diagonal volatility models to avoid overparameterization and argue that the variances and the covariances rely more on its own past than the history of other variances or covariances. In the empirical study developed below, we find that the diagonal models achieve a comparable performance with unrestricted ones, while being much more parsimonious and requiring far less computing time. Notice however, the asymptotic theory developed below is also valid for unrestricted matrices $A_j$'s, $B_i$'s and $C$.
 
 The estimation of the parameters $\theta = (\nu, \text{diag}(C)',\text{diag}(B_i)_{1 \le i \le p}',
 \text{diag}(A_j)_{1 \le j \le q}')'$ of the diagonal CAW$(p,q)$ model is carried out by maximizing the log-likelihood function using the Broyden-Fletcher-Goldfarb-Shanno (BFGS) optimization procedure. Positivity of the diagonal elements $C_{kk}$, $A_{11,j}$ and $B_{11,i}$ are enforced, where $1 \le k \le r$. The log-likelihood function is
 \begin{eqnarray}
 	\mathcal{L}(\theta) &=& \frac{1}{T}\sum_{t=1}^{T}\{-\frac{\nu r}{2}\text{ln}(2) - \frac{r(r-1)}{4}\text{ln}(\pi)-\sum_{i=1}^{r}\text{ln}\Gamma(\frac{\nu+1-i}{2}) -\frac{\nu}{2} \text{ln}|\frac{\mathbf{S}_f(t)}{\nu}| \nonumber \\
 	&&  + (\frac{\nu-r-1}{2})\text{ln}|\mathbf{\tilde{\Sigma}}_f(t)|-\frac{1}{2}\text{tr}(\nu\mathbf{S}_f(t)^{-1}\mathbf{\tilde{\Sigma}}_f(t))\}
 \end{eqnarray}
 In practice, initial values for $\mathbf{S}_f(t)$ are needed to run the maximization of the log-likelihood function. For example, if the order $(p,q)=(2,2)$ is used, then the initial values $\mathbf{S}_f(1)$ and $\mathbf{S}_f(2)$ are needed. In the empirical analysis of this paper, we take $\mathbf{S}_f(1)=\mathbf{\hat{\Sigma}}_f(1)$ and $\mathbf{S}_f(2)=\mathbf{\hat{\Sigma}}_f(2)$ as $\mathbf{S}_f(t)$ is the conditional expectation of $\mathbf{\tilde{\Sigma}}_f(t)$, for any $t$.\\
 
 \section{Asymptotic theory}	
 \noindent
 Given a $d$-dimensional vector $\mathbf{x} = (x_1,\cdots,x_d)'$ and a $d \times d$ matrix $\mathbf{U} = (U_{ij})$, define vector and matrix norms as 
 \begin{equation}
 	||\mathbf{x}||_2 = (\sum_{i=1}^{d}|x_i|^2)^{1/2}, \quad ||\mathbf{U}||_2 = \text{sup}\{||\mathbf{Ux}||_2,||\mathbf{x}||_2=1\}
 \end{equation}
 In fact, $||\mathbf{U}||_2$ is the spectral norm, equal to the square root of the largest eigenvalue of $\mathbf{U}'\mathbf{U}$. In addition, define the Frobenius norm of the $d \times d$ matrix $\mathbf{U} = (U_{ij})$ as
 \begin{equation}
 	||\mathbf{U}||_F = \sqrt{\sum_{i=1}^{d}\sum_{j=1}^{d}|U_{ij}|^2}
 \end{equation}
 
 \noindent
 The asymptotic theory below uses the following assumptions.\\
 (A1) All row vectors of $\mathbf{A}'$ and $\mathbf{\Sigma}_0$ in the factor model (4) satisfy the sparsity condition (13) below. We say that a $d$-dimensional vector $\mathbf{x}=(x_1,\cdots,x_d)'$ is sparse if
 \begin{equation}
 	\sum_{j=1}^{d} |x_i|^{\delta} \le C\pi(d),
 \end{equation}
 where $0 \le \delta <1$, $\pi(d)$ is a deterministic function of $d$ that grows slowly in $d$, such as $\pi(d) = 1$ or $\text{ln}(d)$, and $C$ is a positive constant.\\
 (A2) The factor model (4) has $r$ fixed factors, with $\mathbf{A}'\mathbf{A}=\mathbf{I}_r$, and matrices $\mathbf{\Sigma}_0$ and $\mathbf{\Sigma}_f$ satisfy
 \begin{eqnarray}
 	||\mathbf{\Sigma}_0||_2 &<& \infty \nonumber \\
 	\max_{1 \le t \le T} |\mathbf{\Sigma}_{f,jj}(t)| &=& \mathcal{O}_p(B(T)), \quad j=1,\cdots,r. 
 \end{eqnarray}
 where $1 \le B(T)=o(T)$.\\
 (A3) $\max_{1 \le t \le T}\parallel \mathbf{\hat{\Sigma}}_x(t) - \mathbf{\Sigma}_x(t) \parallel_2 = \mathcal{O}_p(A(d,T,n))$, for some rate function $A(d,T,n)$ such that $A(d,T,n)B^5(T)=o(1)$ with $B(T)$ defined in (A2).\\
 (A4) $A_j$'s and $B_i$'s are such that the CAW model (9) and (10) are stationary and ergodic.\\
 (A5) The parameter set $\Theta$ for all parameters $A_j$'s, $B_i$'s, $C$ and $\nu$ is compact for the CAW model (9) and (10).\\
 (A6) The Hessian matrix $\partial^2\mathcal{L}(\mathbf{\theta})/\partial\theta_i\partial\theta_j$ converges to some deterministic matrix function $D(\mathbf{\theta})$ as $T$ goes to infinity which is of full rank for all $\mathbf{\theta} \in \Theta$.\\
 
 The first two conditions are from \citet{TWYZ11} which are used to prove the consistency of $\mathbf{\hat{\Sigma}}_f(t)$. In this paper, as we are using the threshold MSRVM estimator, we take $A(d,T,n)=\pi(d)[e_n(d^2T)^{1/\beta}]^{1-\delta}\text{ln}T$ and $B(T) = \text{ln}T$, where $e_n \sim n^{-1/4}$, consistent with \citet{TWC13}. In addition, the constant $\beta$ can be taken large enough (under reasonable moment conditions) so that $A(d,T,n)B^5(T)$ will go to $0$ as $n$, $d$, and $T$ go to infinity. Consequently, we assume here these two values shall converge to $0$ in some sense.\\
 
 \noindent
 \begin{thm}
 	Suppose the models (1), (3) and (4) satisfy Conditions (A1)-(A3). Denote the ordered eigenvalues of $\bar{\mathbf{S}}_x$ by $\lambda_1 \ge \cdots \ge \lambda_d$. Let $\mathbf{a}_1, \cdots, \mathbf{a}_r$ be the eigenvectors of $\bar{\mathbf{S}}_x$ corresponding to the $r$ largest eigenvalues $\lambda_1 , \cdots , \lambda_r$. Also set $\hat{\lambda}_1 \ge \cdots \ge \hat{\lambda}_r$ be the $r$ largest eigenvalues of $\bar{\hat{\mathbf{S}}}_x$ and $\mathbf{\hat{a}}_1, \cdots, \mathbf{\hat{a}}_r$ the corresponding eigenvectors. Let $\mathbf{A}=(\mathbf{a}_1, \cdots, \mathbf{a}_r)$ and $\mathbf{\hat{A}}=(\mathbf{\hat{a}}_1, \cdots, \mathbf{\hat{a}}_r)$. As $n$, $d$, $T$ go to infinity, we have
 	\begin{eqnarray}
 		\mathbf{A}'\mathbf{\hat{A}} - \mathbf{I}_r &=& \mathcal{O}_p(A(d,T,n)B(T)), \nonumber \\
 		\mathbf{\hat{\Sigma}}_f(t) - \mathbf{\Sigma}_f(t) - \mathbf{A}'\mathbf{\Sigma}_0\mathbf{A} &=&  \mathcal{O}_p(A^{1/2}(d,T,n)B^{3/2}(T)).
 	\end{eqnarray}
 \end{thm}
 
 \noindent
 \begin{thm}
 	Suppose that $\mathbf{\hat{\theta}}$ is the maximized log-likelihood estimator of $\mathbf{\theta}$ based on the data $\mathbf{\hat{\Sigma}}_f(t)$ from the CAW model and $\mathbf{\tilde{\theta}}$ is the maximized log-likelihood estimator based on the true data $\mathbf{\tilde{\Sigma}}_f(t)$ from the same CAW model. Then under Conditions (A1)-(A5), 
 	\begin{equation}
 		\mathbf{\hat{\theta}}-\mathbf{\tilde{\theta}} = \mathcal{O}_p(A^{1/2}(d,T,n)B^{5/2}(T))
 	\end{equation}
 \end{thm}
 
\section{Data analysis 1}
\noindent
We apply the proposed methodology to two datasets. In this section, we focus on comparing the VAR and CAW models, using 5-minutes intraday data of 30 stocks traded in the US stock market over a period of 103 days where the data are obtained directly from Bloomberg.\\

\subsection{Data description}
\noindent
We use $30$ stocks traded at the New York Stock Exchange, which consist of 27 components of Dow Jones Index : 3M (MMM), American Express (AXP), AT$\&$T (T), Boeing (BA), Caterpillar (CAT), Chevron (CVX), Coca-Cola (KO), Dupont (DD), ExxonMobil (XOM), General Electric (GE), Goldman Sachs (GS), The Home Depot (HD), IBM (IBM), Johnson $\&$ Johnson (JNJ), JPMorgan Chase (JPM), McDonald's (MCD), Merck (MRK), Nike (NKE), Pfizer (PFE), Procter $\&$ Gamble (PG), Travelers (TRV), UnitedHealth Group (UNH), United Technologies (UTX), Verizon (VZ), Visa (V), Wal-Mart (WMT), Walt Disney (DIS) and three former components of Dow Jones Index: Honeywell (HON), Citigroup (C) and American International Group (AIG). The daily realized covariance matrix is computed as $\mathbf{\hat{\Sigma}}_x(t)=\sum_{j=1}^{M} y_{t,j}y_{t,j}'$, where $y_{t,j}$ is the vector of log-returns for the $30$ stocks computed for the $j$th 5-minute interval of trading day $t$ except the first and last half hour. The sample period starts at May 3, 2013 and ends on September 30, 2013, totally 103 trading days. (Here, we exclude July 4 due to the incompleteness of data). This generates a series of 103 matrices of $\mathbf{\hat{\Sigma}}_x(t)$, which are 30 by 30.

Descriptive statistics of the realized variances and covariances are provided in Table $\ref{t1}$. We only show some entries due to limited space. The following properties are found:
\begin{enumerate}
	\item Among 30 realized variances and 435 covariances, only 3 realized covariances are skewed to the left rather than to the right.
	\item All realized variances and covariances have bigger kurtosis than that of the normal distribution except for 2 realized covariances.
\end{enumerate}

For the next two subsections, we show the estimated results for both the diagonal CAW and VAR models when the first $k=98$ days are treated as data points.\\

\subsection{Model fitting}
\noindent
The eigenvalues of the sample variance matrix $\mathbf{\bar{\hat{S}}}_x$ are evaluated, and are shown in Figure $\ref{f1}$. The plots show that the three largest eigenvalues are much larger than the others, which indicates that the factor number $r=3$ is appropriate.

Let $\mathbf{\hat{A}}$ be the eigenvector of $\mathbf{\bar{\hat{S}}}_x$ corresponding to the three largest eigenvalues. We then calculate the factor volatility matrices $\mathbf{\hat{\Sigma}}_f(t)$, which are 3 by 3. Figure $\ref{f2}$ shows time series plots of the variances and covariances of the factor covariance matrices. 

Then, we fit the diagonal CAW model to the matrix series. Different orders are used to make comparison, namely $(p,q)=(0,1)$, $(p,q)=(1,1)$, $(p,q)=(1,2)$, $(p,q)=(2,1)$ and $(p,q)=(2,2)$. \\

\begin{remark}
	The initial value for each numerical optimization is chosen randomly and we repeat 160 times to choose the one with the largest log-likelihood value.
\end{remark}

\subsection{VAR model}
\noindent
For comparison purpose, we also fit the VAR model, advocated in \citet{TWYZ11}, to the vectorized factor covariance matrix $\mathbf{\hat{\Sigma}}_f(t)$, which is a vector with 6 entries. Using the package "vars" in R, we select the VAR(1) model as all the model selection criteria, namely Akaike information criterion (AIC), Hannan-Quinn (HQ), Schwarz criterion (SC) and final prediction error (FPE), choose the order 1, as shown in Table $\ref{t2}$. 

The model is
\begin{displaymath}
	\text{vech}\{\tilde{\mathbf{\Sigma}}_f(t)\} = \mathbf{A}_0 + \mathbf{A}_1 \text{vech}\{\tilde{\mathbf{\Sigma}}_f(t-1)\} + \mathbf{e}(t)
\end{displaymath}
where $\mathbf{A}_0$ is a $6$-dimensional vector, $\mathbf{A}_1$ is a $6 \times 6$ square matrix, and $\mathbf{e}(t)$ is a $6$-dimensional vector white noise process with zero mean and finite fourth moments. Both $\mathbf{A}_0$ and $\mathbf{A}_1$ are estimated by the least squares method.

The fitted VAR(1) model has the coefficients:\\

\noindent
$\hat{\mathbf{A}}_0 = 10^{-5}
\begin{pmatrix}
54.49 \\
-1.482 \\
4.937 \\
7.632 \\
0.164 \\
10.27 
\end{pmatrix}$, \quad
$\hat{\mathbf{A}}_1 = 10^{-2}
\begin{pmatrix}
55.8 & -31.2 & 19.6 & -37.4 & -199.1 & -8.0 \\
-4.1 & 20.6 & 3.2 & 7.7 & 45.4 & -2.0 \\
6.8 & 15.3 & 72.8 & -6.5 & -230.9 & -0.1 \\
0.5 & -26.5 & -84.6 & 4.4 & 283.8 & -6.0 \\
-0.6 & 2.5 & 2.2 & 1.0 & -3.7 & -0.3 \\
10.1 & 43.5 & -17.4 & 11.1 & -207.5 & 13.4
\end{pmatrix}$\\

\noindent
The fitted models are shown in Figure $\ref{f3}$ which indicates that the VAR(1) provides adequate fit to the data. For each of the 6 time series, we show the in-sample fit in the upper panel, where the solid line stands for the real data and the dashed line is the fitted series. The residuals are shown in the middle panel while the ACF and PACF of residuals are displayed in the lower panel. The ACF and PACF plots show that the residuals are time-uncorrelated.\\

\subsection{Performance comparison in out-of-sample forecasting}
\noindent
We compare the out-of-sample one-day-ahead forecast performance of the diagonal CAW and VAR models. The one-day-ahead realized covariance is calculated by: 1. Obtain the one-day-ahead factor covariance matrix by conditional expectation; 2. Plug the forecast factor covariance matrix into the factor model (4) to get the predicted realized covariance matrix.

The predictive accuracy is measured with both the Frobenius norm and the spectral norm. We take the first $k$ days as data and forecast the next day, where $k=80,\cdots,98$. Every model is re-estimated and the new forecasts are generated based upon the new parameter estimates. Then, we take the average of errors during 19 periods to do the comparison. Moreover, the errors of the inverse of matrices are also compared.

Table $\ref{t3}$ contains the results of the prediction accuracy of the two models using different norms. The main findings are as follows.
\begin{enumerate}
	\item The CAW models with order $(p,q)=(1,1)$ performs the best as it has the smallest error under both Frobenius and spectral norms. In general, All CAW models have similar performance except the ones with order $(p,q)=(0,1)$.
	\item The CAW models have similar performance with that of the VAR(1) model except the one with order $(p,q)=(0,1)$.
	\item However, the diagonal CAW model needs far less parameters than the VAR model does. Here, the best CAW model with order $(p,q)=(1,1)$ only needs 10 parameters, just a quarter of the number of parameters that the VAR(1) model needs.
	\item As $d$ goes up, we can imagine that $r$ will become larger; in this situation, the parameters that we need for the diagonal CAW model will be even far less than that we need for the VAR model.
\end{enumerate}

\begin{remark}
	\noindent	
	\begin{enumerate}
		\item We do predictions for 2 to 5 days ahead in addition to the one-day forecast. We find that the performance of predictions for longer horizons is similar with that for the one day, in general.
		\item One problem for the VAR model is that it cannot assure that the predicted factor covariance matrix is positive definite. We have checked that in this data analysis, the predicted factor covariance matrices are actually all positive definite. This may be due to the reason that the estimated factor variances are larger than the covariances (in absolute value).
	\end{enumerate}
\end{remark}

\section{Data analysis 2}
\subsection{Data description}
\noindent
We use the same 30 stocks as in the previous section. The raw tick-by-tick trading data are downloaded from TAQ database of Wharton Research Data Service. The data period starts at January 3, 2012 and ends on December 31, 2012, with totally 250 trading days.

We firstly conduct data cleaning with the procedures introduced in \citet{BG06} and \citet{BHLS09}. The steps are the following:\\
\begin{enumerate}
	\item Delete entries with a time stamp outside 9:30 am - 4:00 pm when the exchange is open.
	\item Delete entries with a time stamp inside 9:30 - 10:00 am or 3:30 - 4:00 pm to eliminate the open and end effect of price fluctuation.
	\item Delete entries with the transaction price equal to zero.
	\item If multiple transactions have the same time stamp, use the median price.
	\item Delete entries with prices which are outliers. Let $\{p_i\}_{i=1}^N$ be an ordered tick-by-tick price series. We treat the $i$-th price as an outlier if $|p_i-\bar{p}_i(k)|>3s_i(k)$, where $\bar{p}_i(k)$ and $s_i(k)$ denote the sample mean and sample standard deviation of a neighborhood of $k$ observations around $i$, respectively. For the beginning prices which may not have enough left hand side neighbors, we get $k-i$ neighbors from $i+1$ to $k+1$. Similar procedures are taken for the ending prices.   
\end{enumerate}
Then, we construct the threshold MSRVM estimator based on the cleaned tick-by-tick data following the steps in \citet{TWC13}. We set the threshold to be $5\%$ of the largest of the absolute value of entries in the matrix.

Descriptive statistics of selected realized variances and covariances are provided in Table $\ref{t4}$. In addition, two plots for realized variances and two plots for realized covariances are shown in Figure $\ref{f4}$. We find the following properties:
\begin{enumerate}
	\item All 30 realized variances and 435 covariances are skewed to the right, with mean skewness $1.39$.
	\item All realized variances and covariances have bigger kurtosis than that of the normal distribution, with mean kurtosis $5.92$ showing fat tails.
	\item The realized variances and covariances have significant fluctuations during the year, indicated by the graphs.
\end{enumerate}

For the next two subsections, we show the estimated results for both the diagonal CAW and VAR models when the first $k=220$ days are treated as data points.\\

\subsection{Model fitting}
\noindent
The eigenvalues of the sample variance matrix $\mathbf{\bar{\hat{S}}}_x$ are evaluated, and are shown in Figure $\ref{f5}$. We choose four factors for the model, as there is a dicernable drop between the fourth and the fifth eigenvalues, though the second to fourth eigenvalues are much less than the biggest one.

Let $\mathbf{\hat{A}}$ be the eigenvectors of $\mathbf{\bar{\hat{S}}}_x$ corresponding to the four largest eigenvalues. We calculate the factor volatility matrices $\mathbf{\hat{\Sigma}}_f(t)$, which are 4 by 4. Then, we fit the diagonal CAW model to the matrix series. Different orders are used to make comparison, namely $(p,q)=(0,1)$, $(p,q)=(1,1)$, $(p,q)=(1,2)$, $(p,q)=(2,1)$ and $(p,q)=(2,2)$. For every estimation, we randomly choose 60 initial values for each optimization of the log-likelihood and choose the one with the largest log-likelihood value to give the estimated parameters.\\

\subsection{VAR model}
\noindent
For comparison purpose, we also fit the VAR model to the vectorized factor covariance matrix $\mathbf{\hat{\Sigma}}_f(t)$, which is a vector with 10 entries. Again, using the package "vars" in R, we select the VAR(1) model as all the model selection criteria AIC, HQ, SC and FPE choose the order 1, as shown in Table $\ref{t5}$.

The model is
\begin{displaymath}
	\text{vech}\{\tilde{\mathbf{\Sigma}}_f(t)\} = \mathbf{A}_0 + \mathbf{A}_1 \text{vech}\{\tilde{\mathbf{\Sigma}}_f(t-1)\} + \mathbf{e}(t)
\end{displaymath}
where $\mathbf{A}_0$ is a $10$-dimensional vector, $\mathbf{A}_1$ is a $10 \times 10$ square matrix, and $\mathbf{e}(t)$ is a $10$-dimensional vector white noise process with zero mean and finite fourth moments. Both $\mathbf{A}_0$ and $\mathbf{A}_1$ are estimated by the least squares method.

We denote the estimator of $\mathbf{A}_1$ by $\mathbf{\hat{A}}_1$. We find that $|\mathbf{\hat{A}}_1| = -1.03 \times 10^{-7}$ and the biggest absolute eigenvalues is $0.4906 < 1$ which ensures the stationarity of the VAR model.

In addition, the sparsity of $\mathbf{\hat{A}}_1$ is checked. We compute the values $\frac{1}{10^2}\sum_{i,j} |a_{i,j}|^m$ for different $0 \le m < 1$, where $\mathbf{\hat{A}}_1 = \{a_{i,j}\}_{1 \le i,j \le 10}$. The following is the result. We can see from Table $\ref{t6}$ that the average absolute value of the entries of $\mathbf{\hat{A}}_1$ is small (less than 1) so that we can consider $\mathbf{\hat{A}}_1$ as almost sparse.\\

\subsection{Performance comparison in out-of-sample forecasting}
\noindent
We compare the out-of-sample one-day-ahead forecast performance of the diagonal CAW and VAR models. The prediction of the one-day-ahead realized covariance is calculated by: 1. Predict the one-day-ahead factor covariance matrix by conditional expectation; 2. Plug the forecast factor covariance matrix into the factor model (4) to get the predicted realized covariance matrix.

The predictive accuracy is measured with both the Frobenius norm and the spectral norm. We take the first $k$ days as data and forecast the next day, where $k=220,\cdots,240$. Every model is re-estimated and the new forecasts are generated based upon the new parameter estimates. Then, we take the average of errors during 21 periods to do the comparison. Moreover, the errors of the inverse of matrices are also compared.

Table $\ref{t7}$ contains the results of the prediction error of two models using different norms. The main findings are as follows.
\begin{enumerate}
	\item The diagonal CAW models with order $(p,q)=(1,1)$ performs the best among the CAW models as it has the smallest error under both Frobenius and spectral norms. In general, all CAW models have similar performance except the ones with order $(p,q)=(0,1)$. The result also indicates the possibility of over-parameterization with orders larger than $(1,1)$.
	\item The diagonal CAW models have slightly worse, but comparable performance with that of the VAR(1) model except the one with order $(p,q)=(0,1)$.
	\item The diagonal CAW model needs far less parameters than the VAR model does. Here, the best CAW model with order $(p,q)=(1,1)$ only needs 13 parameters, nearly a tenth of the number of parameters that VAR(1) model needs.
	\item We do predictions for 2 to 5 days ahead in addition to one-day forecast. We find that the performance of predictions of longer horizons is similar with that for the one day, in general.
\end{enumerate}

\section{Conclusions}
\noindent
In the literature, most models dealing with the realized covariance matrix focus on small number of assets, which become infeasible when the dimension is large. In order to solve the problem, we propose a factor model with diagonal CAW model fitted to the factor covariance matrix. Our model performs comparably with the VAR model while requiring far less parameters. For example, in the second data analysis, the CAW model with order $(p,q)=(1,1)$ performs similarly with the VAR model, measured both in the Frobenius norm and the spectral norm, but only needs nearly one tenth of the number of parameters of the latter. In addition, the model ensures the positive definiteness for the predicted covariance matrices. Diagnostic tests for the proposed model is worth considering in a future study.\\

\newpage

\begin{table}[H]
	\centering
	\caption{Descriptive Statistics for Selected Realized Variances and Covariances \label{t1}}
	\caption*{We report the descriptive statistics of the realized variances and covariances of the first dataset, namely mean, maximum, minimum, standard deviation, skewness and kurtosis. We only show some entries due to limited space.}
	\begin{tabular}{|l*{6}{c}|}
		\hline
		Stock              & Mean & Maximum & Minimum & SD & Skewness  & Kurtosis  \\
		& $*10^{-5}$ & $*10^{-4}$ & $*10^{-5}$ & $*10^{-5}$ & & \\
		\hline
		& \multicolumn{6}{c|}{Realized Variance} \\
		\hline
		MMM & 3.25 & 1.27 & 0.63 & 1.96 & 1.85 & 4.84   \\
		AXP & 6.90 & 6.05 & 2.03 & 7.02 & 4.74 & 31.0   \\
		T   & 4.85 & 1.36 & 1.29 & 2.60 & 1.15 & 0.85   \\
		BA  & 8.60 & 28.5 & 1.31 & 27.8 & 9.61 & 92.8   \\
		CAT & 6.39 & 3.70 & 1.99 & 4.82 & 4.43 & 23.8   \\
		\hline
		& \multicolumn{6}{c|}{Realized Covariance} \\
		\hline
		MMM-AXP & 2.17 & 1.23 & -0.25 & 2.03 & 2.13 & 6.27   \\
		MMM-T   & 1.40 & 0.77 & -0.37 & 1.48 & 1.86 & 3.80   \\
		MMM-BA  & 1.94 & 0.95 & -2.33 & 1.85 & 1.36 & 2.67   \\
		MMM-CAT & 2.11 & 1.08 & -0.03 & 1.72 & 2.23 & 7.01   \\
		AXP-T   & 1.96 & 0.96 & -0.66 & 1.95 & 1.70 & 3.15   \\
		AXP-BA  & 2.18 & 1.17 & -7.27 & 2.39 & 0.45 & 3.36   \\
		AXP-CAT & 2.24 & 1.14 & -0.69 & 2.11 & 1.80 & 4.37   \\
		T-BA    & 1.38 & 0.86 & -1.47 & 1.59 & 1.95 & 5.33   \\
		T-CAT   & 1.41 & 0.84 & -1.13 & 1.61 & 1.80 & 3.95   \\
		BA-CAT  & 2.12 & 0.98 & -0.32 & 1.76 & 1.65 & 3.98   \\ 
		\hline		
	\end{tabular}
	
\end{table}

\newpage
\begin{table}[H]
	\centering
	\caption{The Selection of the Order of VAR Model \label{t2}}
	\caption*{We fit the VAR model to the vectorized factor covariance matrix $\mathbf{\hat{\Sigma}}_f(t)$, which is a vector with 6 entries. Using the package "vars" in R, we select the VAR(1) model by Akaike information criterion (AIC), Hannan Quinn (HQ), Schwarz criterion (SC) and final prediction error (FPE). }
	\begin{tabular}{|l*{5}{c}|}
		\hline
		Order             & 1 & 2 & 3 & 4  & 5  \\
		\hline
		AIC(n) $*10^{2}$ & -1.107 & -1.102 & -1.098 & -1.094 & -1.090    \\
		HQ(n) $*10^{2}$ & -1.102 & -1.093 & -1.084 & -1.077 & -1.068 \\
		SC(n) $*10^{2}$ & -1.094 & -1.079 & -1.065 & -1.051 & -1.036   \\
		FPE(n) $*10^{-48}$ & 0.830 & 1.313 & 2.181 & 3.268 & 5.772   \\
		\hline
	\end{tabular}
\end{table}

\newpage
\begin{table}[H]
	\centering
	\caption{Forecast errors for CAW and VAR models using different norms  \label{t3}}
	\caption*{We report the results of the prediction accuracy of the two models using different norms. In addition, the prediction accuracy of the inverse of the matrices is shown as well. Here, FN is for the Frobenius norm and SN for the spectral norm. }
	\begin{normalsize}
		\begin{tabular}{|l*{5}{c}|}
			\hline
			& \multicolumn{5}{c|}{CAW} \\
			\hline
			Order& Number of Parameters & FN & SN & FN & SN \\
			&& & &  for Inverse &  for Inverse\\
			&& $*10^{-4}$ &  $*10^{-4}$ &  $*10^{5}$ &  $*10^{5}$ \\
			\hline
			$p=0,q=1$ & 7 & 6.21 & 5.66 & 6.30 & 3.84 \\
			$p=1,q=1$ & 10 & 4.77 & 3.95 & 6.31 & 3.84  \\
			$p=1,q=2$ & 13 & 4.92 & 4.01 & 6.31 & 3.84   \\
			$p=2,q=1$ & 13 & 4.96 & 4.01 & 6.31 & 3.84  \\
			$p=2,q=2$ & 16 & 4.98 & 4.03 & 6.31 & 3.84 \\
			\hline
			& \multicolumn{5}{c|}{VAR} \\
			\hline
			& Number of Parameters & FN & SN & FN & SN \\
			&& & &  for Inverse &  for Inverse\\
			&& $*10^{-4}$ &  $*10^{-4}$ &  $*10^{5}$ &  $*10^{5}$ \\
			\hline
			VAR(1) & 42 & 4.81 & 3.99 & 6.31 & 3.84 \\
			\hline
		\end{tabular}
	\end{normalsize}
	
\end{table}

\newpage
\begin{table}[H]
	\centering
	\caption{Descriptive Statistics for Some Selected Realized Variances and Covariances \label{t4}}
	\caption*{We report the descriptive statistics of the realized variances and covariances of the second dataset, namely mean, maximum, minimum, standard deviation, skewness and kurtosis. We only show some entries due to limited space.}
	\begin{tabular}{|l*{6}{c}|}
		\hline
		Stock              & Mean & Maximum & Minimum & SD & Skewness  & Kurtosis  \\
		& $*10^{-5}$ & $*10^{-4}$ & $*10^{-5}$ & $*10^{-5}$ & & \\
		\hline
		& \multicolumn{6}{c|}{Realized Variance} \\
		\hline
		AIG & 17.4 & 9.16 & 3.05 & 11.2 & 2.47 & 13.1   \\
		AXP & 6.12 & 6.38 & 1.45 & 4.71 & 7.67 & 91.5   \\
		BA  & 5.46 & 2.19 & 1.11 & 3.25 & 1.70 & 7.33   \\
		C   & 16.7 & 9.62 & 2.80 & 10.9 & 2.56 & 15.3   \\
		\hline
		& \multicolumn{6}{c|}{Realized Covariance} \\
		\hline
		AIG-AXP & 4.30 & 1.65 & -1.99 & 3.13 & 1.00 & 3.76   \\
		AIG-BA  & 3.32 & 1.32 & -2.74 & 2.82 & 1.03 & 3.93   \\
		AIG-C   & 7.72 & 4.54 & -1.47 & 5.99 & 1.85 & 9.37   \\
		AXP-BA  & 2.41 & 1.01 & -0.59 & 1.95 & 1.26 & 4.37   \\
		AXP-C   & 5.03 & 2.24 & -0.75 & 3.53 & 1.61 & 6.84   \\
		BA-C    & 3.69 & 1.75 & -1.40 & 3.13 & 1.66 & 6.30   \\
		\hline
	\end{tabular}
\end{table}

\newpage
\begin{table}[H]
	\centering
	\caption{The Selection of the Order of VAR Model \label{t5}}
	\caption*{We fit the VAR model to the vectorized factor covariance matrix $\mathbf{\hat{\Sigma}}_f(t)$, which is a vector with 10 entries. Using the package "vars" in R, we select the VAR(1) model as all the model selection criteria, namely Akaike information criterion (AIC), Hannan Quinn (HQ), Schwarz criterion (SC) and final prediction error (FPE), choose the order 1.}
	\begin{tabular}{|l*{5}{c}|}
		\hline
		Oder             & 1 & 2 & 3 & 4  & 5  \\
		\hline
		AIC(n) $*10^{2}$ & -1.949 & -1.943 & -1.937 & -1.933 & -1.932    \\
		HQ(n) $*10^{2}$ & -1.939 & -1.925 & -1.910 & -1.898 & -1.888 \\
		SC(n) $*10^{2}$ & -1.924 & -1.897 & -1.871 & -1.846 & -1.824   \\
		FPE(n) $*10^{-84}$ & 0.234 & 0.416 & 0.796 & 1.292 & 1.751   \\
		\hline
	\end{tabular}
\end{table} 

\newpage
\begin{table}[H]
	\centering
	\caption{Sparsity of $\mathbf{\hat{A}}_1$ \label{t6}}
	\caption*{We compute the values $\frac{1}{10^2}\sum_{i,j} |a_{i,j}|^m$ for different $0 \le m < 1$, where $\mathbf{\hat{A}}_1 = \{a_{i,j}\}_{1 \le i,j \le 10}$.}
	\begin{normalsize}
		\begin{tabular}{|l*{7}{c}|}
			\hline
			$m$ & $0$ & $0.05$ & $0.1$ & $0.15$ & $0.2$ & $0.25$ & $0.3$ \\
			\hline
			$\frac{1}{10^2}\sum_{i,j} |a_{i,j}|^m$ & $1$ & $0.8686$ & $0.7586$ & $0.666$ & $0.5877$ & $0.5212$ & $0.4646$ \\
			\hline
		\end{tabular}
	\end{normalsize}
\end{table}

\newpage
\begin{table}[H]
	\centering
	\caption{Forecast errors for CAW and VAR models using different norms  \label{t7}}
	\caption*{We report the results of the prediction accuracy of the two models using different norms. In addition, the prediction accuracy of the inverse of the matrices is shown as well. Here, FN is for the Frobenius norm and SN for the spectral norm.}
	\begin{normalsize}
		\begin{tabular}{|l*{5}{c}|}
			\hline
			& \multicolumn{5}{c|}{CAW} \\
			\hline
			Order& Number of Parameters & FN & SN & FN & SN \\
			&& & &  for Inverse &  for Inverse\\
			&& $*10^{-4}$ &  $*10^{-4}$ &  $*10^{5}$ &  $*10^{5}$ \\
			\hline
			$p=0,q=1$ & 9 & 6.10 & 5.58 & 7.30 & 4.31 \\
			$p=1,q=1$ & 13 & 5.27 & 4.80 & 7.31 & 4.31  \\
			$p=1,q=2$ & 17 & 5.28 & 4.81 & 7.31 & 4.31   \\
			$p=2,q=1$ & 17 & 5.28 & 4.83 & 7.31 & 4.31 \\
			$p=2,q=2$ & 21 & 5.28 & 4.82 & 7.31 & 4.31 \\
			\hline
			& \multicolumn{5}{c|}{VAR} \\
			\hline
			& Number of Parameters & FN & SN & FN & SN \\
			&& & &  for Inverse &  for Inverse\\
			&& $*10^{-4}$ &  $*10^{-4}$ &  $*10^{5}$ &  $*10^{5}$ \\
			\hline
			VAR(1) & 110 & 5.18 & 4.76 & 7.31 & 4.30 \\
			\hline
		\end{tabular}
	\end{normalsize}
	
\end{table}

\newpage
\begin{figure}[H]
	\centering
	\includegraphics[width=15cm]{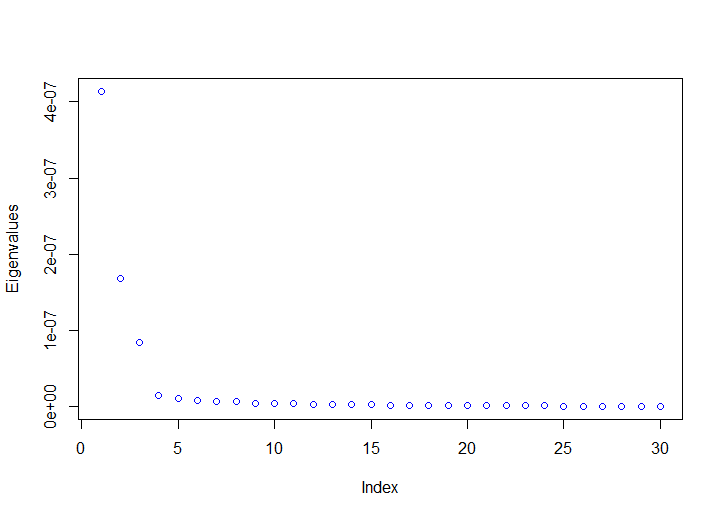}
	\caption{Plots of the eigenvalues of $\mathbf{\bar{\hat{S}}}_x$ for the dataset when $k=98$. }
	\label{f1}
\end{figure}

\newpage
\begin{figure}[H]
	\centering
	
	\includegraphics[width=16cm]{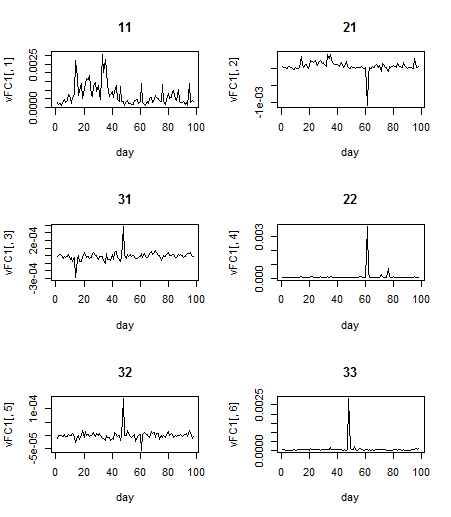}
	\caption{Time series of variances and covariances for factor covariance matrices. }
	\label{f2}

\end{figure}

\newpage

\begin{figure}[H]
	\centering
	\subfloat{\includegraphics[width=11.5cm]{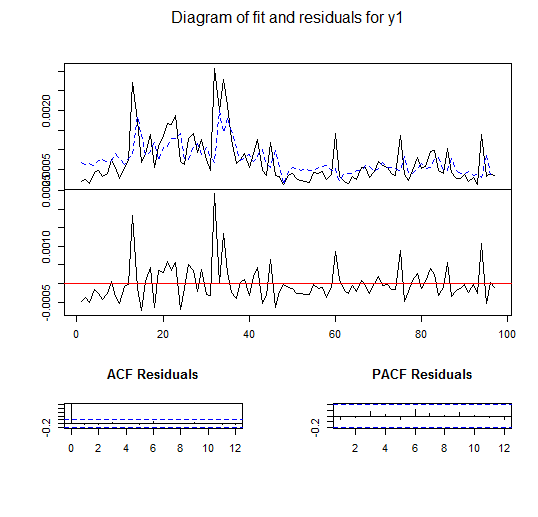}}\,
	\subfloat{\includegraphics[width=11.5cm]{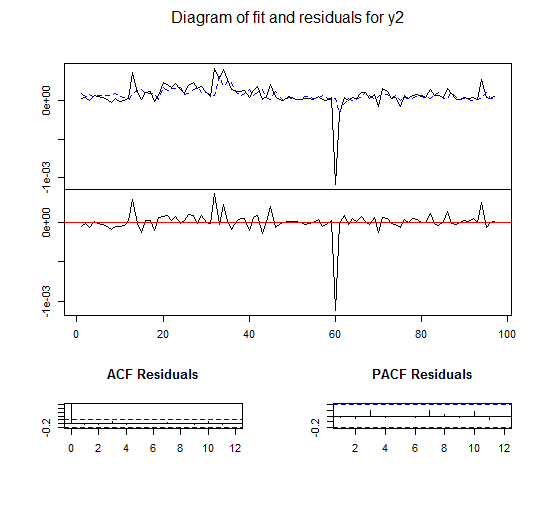}}	
\end{figure}

\begin{figure}[H]
	\ContinuedFloat 
	\centering
	\subfloat{\includegraphics[width=11.5cm]{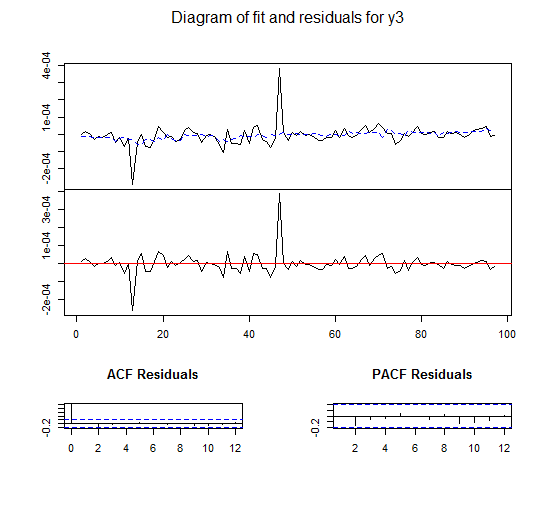}}\,
	\subfloat{\includegraphics[width=11.5cm]{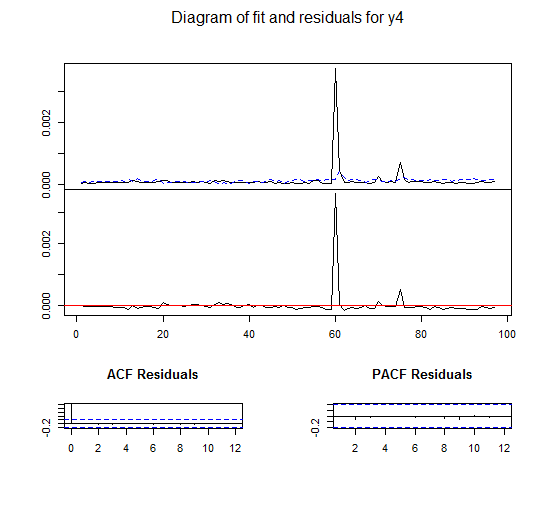}}
\end{figure}

\setcounter{figure}{4}
\begin{figure}[H]
	\ContinuedFloat
	\centering
	\subfloat{\includegraphics[width=11.5cm]{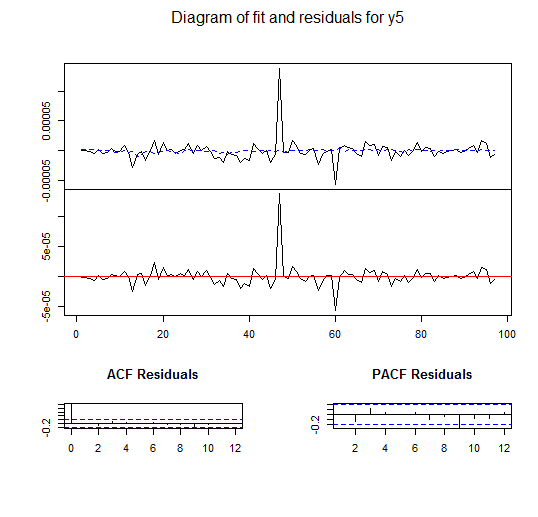}}\,
	\subfloat{\includegraphics[width=11.5cm]{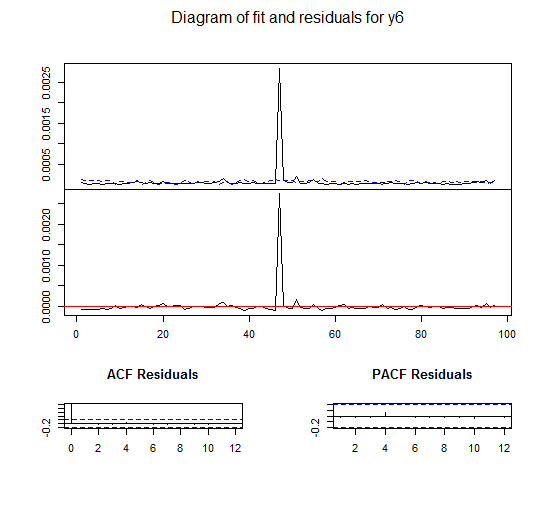}}
	\caption{The plot of fitted VAR(1) model, residuals of fitted VAR(1) model,  and ACF and PACF of the residuals. }
	\label{f3}

\end{figure}

\newpage
\setcounter{figure}{5}
\begin{figure}[H]
	\centering
	\includegraphics[width=15cm]{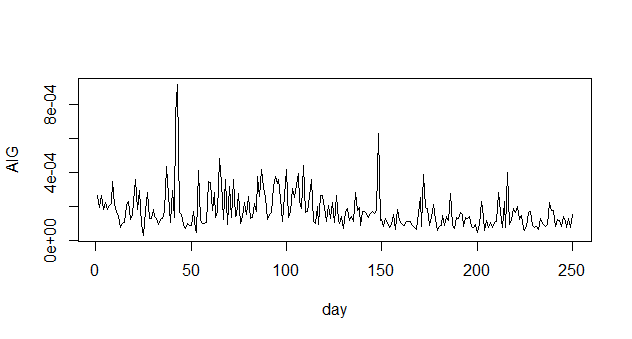}\,
	\includegraphics[width=15cm]{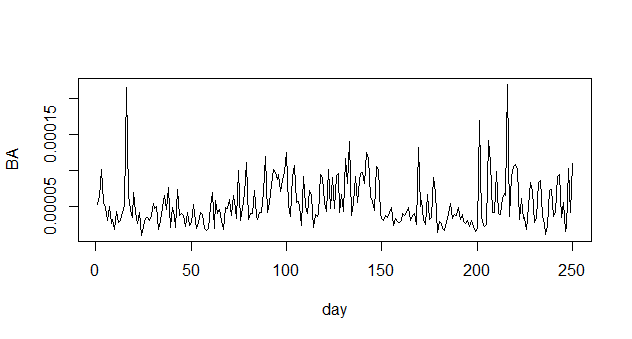}
\end{figure}
\begin{figure}[H] 
	\ContinuedFloat
	\centering   
	\includegraphics[width=15cm]{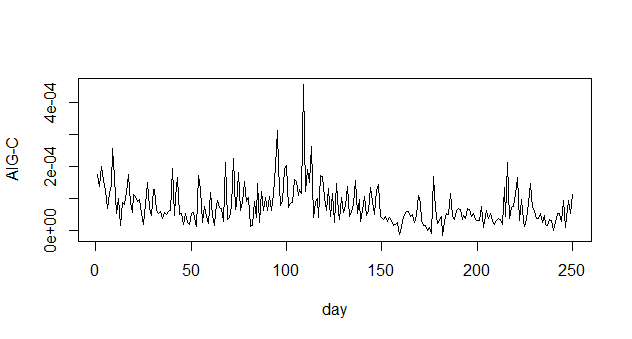}\,
	\includegraphics[width=15cm]{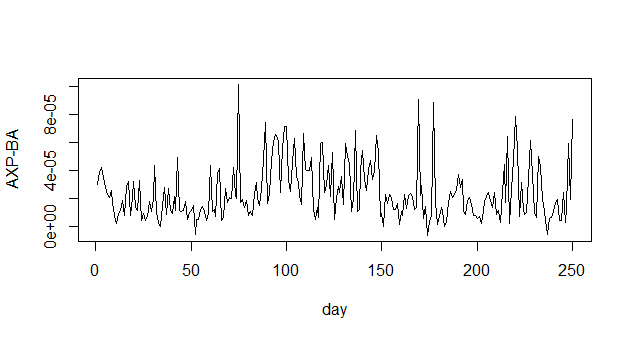}
	\caption{The plot of selected realized variances and covariances. }
	\label{f4}

\end{figure}

\newpage
\setcounter{figure}{4}
\begin{figure}[H]
	\centering
	
	\includegraphics[width=15cm]{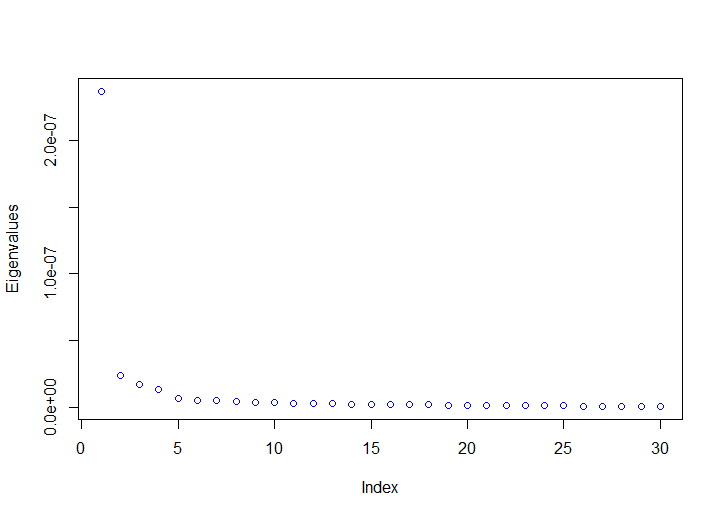}	
	\caption{Plots of the eigenvalues of $\mathbf{\bar{\hat{S}}}_x$ for the dataset when $k=220$.}
	\label{f5}
\end{figure}

\newpage

\appendix
\section{Proof of theorems}

\noindent
For convenience, we denote $A(d,T,n)$ and $B(T)$ by $A$ and $B$ in the proof part.\\

\noindent
\begin{proof}[Proof of Theorem 1]
	Following Theorem 1 in \citet{TWYZ11}, we can easily show that $\parallel \mathbf{\bar{\hat{S}}}_x - \bar{\mathbf{S}}_x \parallel_2 = \mathcal{O}_p(AB)$.\\
	
	\noindent
	Then, we claim that 
	\begin{equation}
		max_{1\le j \le r}\mid \hat{\lambda}_j - \lambda_j \mid = \mathcal{O}_p(AB)
	\end{equation}
	\begin{equation}
		max_{1\le j \le r}\parallel \hat{\mathbf{a}}_j - \mathbf{a}_j \parallel_2 = \mathcal{O}_p(AB)
	\end{equation}
	\noindent
	Let $P = \mathbf{\bar{\hat{S}}}_x - \bar{\mathbf{S}}_x$ with ordered eigenvalues $\rho_1 \ge \cdots \ge \rho_d$. Then, we have $\rho_d \le \hat{\lambda}_j - \lambda_j \le \rho_1$. As a result,
	\begin{equation}
		\mid \hat{\lambda}_j - \lambda_j \mid \le max(\mid \rho_1 \mid, \mid \rho_d \mid) = \parallel \mathbf{\bar{\hat{S}}}_x - \bar{\mathbf{S}}_x \parallel_2
	\end{equation}
	which proves equation (A.1). The equation (A.2) follows from Theorem 1 in \citet{BL08a} and the same argument in the proof of Theorem 5 in the same paper.\\
	
	\noindent
	For the $j$-th diagonal entry of $\mathbf{A}'\mathbf{\hat{A}} - \mathbf{I}_r$,
	\begin{equation}
		\mathbf{a}_j'\hat{\mathbf{a}}_j-1 = -\parallel \hat{\mathbf{a}}_j - \mathbf{a}_j \parallel_2^2/2 = \mathcal{O}_p(A^2B^2) \le \mathcal{O}_p(AB)
	\end{equation}
	as we assume $AB$ goes to zero. For off-diagonal entry $(k,j)$ $(k \ne j)$,
	\begin{eqnarray}
		\mid \mathbf{a}_k'\hat{\mathbf{a}}_j \mid &=& \mid \mathbf{a}_k'(\hat{\mathbf{a}}_j-\mathbf{a}_j) \mid \nonumber \\
		&\le&  \parallel \mathbf{a}_k' \parallel_2 \parallel \hat{\mathbf{a}}_j - \mathbf{a}_j \parallel_2 = \parallel \hat{\mathbf{a}}_j - \mathbf{a}_j \parallel_2 \nonumber \\
		&=& \mathcal{O}_p(AB) 
	\end{eqnarray}
	which proves the first result.\\
	
	\noindent
	To prove the second result, we separate the left hand side of the equation into three parts:
	\begin{eqnarray}
		&& \mathbf{\hat{\Sigma}}_f(t) - \mathbf{\Sigma}_f(t) - \mathbf{A}'\mathbf{\Sigma}_0\mathbf{A} \nonumber \nonumber \\
		&=& \mathbf{\hat{A}}'[\mathbf{\hat{\Sigma}}_x(t)-\mathbf{\Sigma}_x(t)]\mathbf{\hat{A}} + \mathbf{\hat{A}}'\mathbf{\Sigma}_x(t)\mathbf{\hat{A}} - \mathbf{\Sigma}_f(t) - \mathbf{A}'\mathbf{\Sigma}_0\mathbf{A} \nonumber \\
		&=& \mathbf{\hat{A}}'[\mathbf{\hat{\Sigma}}_x(t)-\mathbf{\Sigma}_x(t)]\mathbf{\hat{A}} + [(\mathbf{A}'\mathbf{\hat{A}})'\mathbf{\Sigma}_f(t) \mathbf{A}'\mathbf{\hat{A}} - \mathbf{\Sigma}_f(t)] \nonumber \\
		&& + [\mathbf{\hat{A}}'\mathbf{\Sigma}_0\mathbf{\hat{A}} - \mathbf{A}'\mathbf{\Sigma}_0\mathbf{A}] 
	\end{eqnarray}.
	
	\noindent
	For the first term on the right-hand side of equation (A.6), since
	\begin{equation}
		\parallel \mathbf{\hat{A}}' \parallel_2^2,\parallel \mathbf{\hat{A}} \parallel_2^2 =1,
	\end{equation}
	we have
	\begin{eqnarray}
		\parallel \mathbf{\hat{A}}'[\mathbf{\hat{\Sigma}}_x(t)-\mathbf{\Sigma}_x(t)]\mathbf{\hat{A}} \parallel_2 &\le& \parallel \mathbf{\hat{A}}'\parallel_2 \parallel \mathbf{\hat{\Sigma}}_x(t)-\mathbf{\Sigma}_x(t)\parallel_2 \parallel \mathbf{\hat{A}} \parallel_2 \nonumber \\
		&=& \mathcal{O}_p(A) 
	\end{eqnarray}
	
	\noindent
	For the second term, from Condition (A2), we know that $\parallel \mathbf{\Sigma}_f(t) \parallel_2 = \mathcal{O}_p(B)$, and we have
	\begin{eqnarray}
		\parallel \mathbf{\hat{A}}-\mathbf{A}\parallel_2^2 &=& \parallel (\mathbf{\hat{A}}'-\mathbf{A}') (\mathbf{\hat{A}}-\mathbf{A})\parallel_2 \nonumber \\
		&\le& 2\parallel \mathbf{A}'\mathbf{\hat{A}}-\mathbf{I}_r\parallel_2 = \mathcal{O}_p(AB) 
	\end{eqnarray}
	As a result,
	\begin{eqnarray}
		\parallel (\mathbf{A}'\mathbf{\hat{A}})'\mathbf{\Sigma}_f(t) \mathbf{A}'\mathbf{\hat{A}} - \mathbf{\Sigma}_f(t) \parallel_2 &=& \parallel \mathbf{\hat{A}}'\mathbf{A}\mathbf{\Sigma}_f(t) \mathbf{A}'\mathbf{\hat{A}} - \mathbf{\hat{A}}'\mathbf{\hat{A}}\mathbf{\Sigma}_f(t) \mathbf{\hat{A}}'\mathbf{\hat{A}} \parallel_2 \nonumber \\
		&\le& \parallel \mathbf{\hat{A}}' \parallel_2 \parallel \mathbf{A}\mathbf{\Sigma}_f(t) \mathbf{A}' - \mathbf{\hat{A}}\mathbf{\Sigma}_f(t)\mathbf{\hat{A}}' \parallel_2 \parallel \mathbf{\hat{A}} \parallel_2 \nonumber \\
		&\le& \parallel (\mathbf{\hat{A}}-\mathbf{A})\mathbf{\Sigma}_f(t) \mathbf{\hat{A}}' + \mathbf{A}\mathbf{\Sigma}_f(t) (\mathbf{\hat{A}}-\mathbf{A})' \parallel_2 \nonumber \\
		&\le&  \parallel \mathbf{\hat{A}}-\mathbf{A} \parallel_2 \parallel \mathbf{\Sigma}_f(t)\parallel_2 (\parallel \mathbf{\hat{A}} \parallel_2+\parallel \mathbf{A} \parallel_2) \nonumber \\
		&=& \mathcal{O}_p(A^{1/2}B^{3/2}) 
	\end{eqnarray}
	
	\noindent
	For the third term, again from Condition (A2), we know that $\parallel \mathbf{\Sigma}_0 \parallel_2$ is bounded, therefore,
	\begin{eqnarray}
		\parallel \mathbf{\hat{A}}'\mathbf{\Sigma}_0 \mathbf{\hat{A}} - \mathbf{A}'\mathbf{\Sigma}_0\mathbf{A} \parallel_2 &=& \parallel (\mathbf{\hat{A}}-\mathbf{A})'\mathbf{\Sigma}_0 \mathbf{\hat{A}} + \mathbf{A}'\mathbf{\Sigma}_0 (\mathbf{\hat{A}}-\mathbf{A}) \parallel_2 \nonumber \\
		&\le& \parallel (\mathbf{\hat{A}}-\mathbf{A})' \parallel_2 \parallel \mathbf{\Sigma}_0 \parallel_2 \parallel \mathbf{\hat{A}} \parallel_2 + \parallel \mathbf{A}' \parallel_2 \parallel \mathbf{\Sigma}_0 \parallel_2 \parallel (\mathbf{\hat{A}}-\mathbf{A}) \parallel_2 \nonumber \\
		&=& \parallel \mathbf{\hat{A}}-\mathbf{A} \parallel_2 \parallel \mathbf{\Sigma}_0\parallel_2 (\parallel \mathbf{\hat{A}} \parallel_2+\parallel \mathbf{A} \parallel_2) \nonumber \\
		&=& \mathcal{O}_p(A^{1/2}B^{1/2}) 
	\end{eqnarray}
	
	\noindent
	From equations (A.8), (A.10) and (A.11), we conclude that 
	\begin{equation}
		\mathbf{\hat{\Sigma}}_f(t) - \mathbf{\Sigma}_f(t) - \mathbf{A}'\mathbf{\Sigma_0}\mathbf{A} = \mathcal{O}_p(A^{1/2}B^{3/2}).
	\end{equation}
	
\end{proof}

\noindent
We need the following three lemmas to prove Theorem 2.\\

\noindent
\begin{lem} $\partial\mathcal{L}(\mathbf{\theta})/\partial\theta_i$ and $ \partial^2\mathcal{L}(\mathbf{\theta})/\partial\theta_i\partial\theta_j$ are uniformly bounded for $\theta \in \Theta$.\\
\end{lem}

\begin{proof}[Proof]
	$\mathbf{S}_f(t) \ge CC' > 0$, so that $\mathbf{S}_f(t)^{-1}>0$. Given $\mathbf{S}_f(0)$, $\mathbf{S}_f(-1)$, $\cdots$, $\mathbf{S}_f(-p+1)$, $\mathbf{S}_f(t)$ is a polynomial of $A_{mn,j},B_{mn,i},C_{mn}$'s, for $1 \le m \le r$ and $1 \le n \le r$, i.e. $\mathbf{S}_f(t) \in \mathcal{C}^{\infty}$. As a result, $\mathbf{S}_f(t)^{-1} \in \mathcal{C}^{\infty}$. We have $\text{ln}(\cdot)$, $\Gamma(\cdot)$ and $\text{tr}(\cdot)$ are all $\in \mathcal{C}^{\infty}$, which indicates $\mathcal{L}(\theta) \in \mathcal{C}^{\infty}$, i.e.  $\partial\mathcal{L}(\mathbf{\theta})/\partial\theta_i$, $ \partial^2\mathcal{L}(\mathbf{\theta})/\partial\theta_i\partial\theta_j \in \mathcal{C}^{\infty}$. Since $\theta \in \Theta$ which is compact, we have the conclusion that $\partial\mathcal{L}(\mathbf{\theta})/\partial\theta_i$ and $ \partial^2\mathcal{L}(\mathbf{\theta})/\partial\theta_i\partial\theta_j$ are bounded on $\Theta$.
\end{proof}. 

\noindent
\begin{lem} If we denote $K(\theta,\mathbf{\tilde{\Sigma}}_f(t)):=\mathcal{\hat{L}}(\theta) - \mathcal{L}(\theta) = \mathcal{O}_p(A^{1/2}B^{5/2})$, then we have 
	\begin{equation}
		\frac{\partial K(\theta,\mathbf{\tilde{\Sigma}}_f(t))}{\partial\theta_i} = \mathcal{O}_p(A^{1/2}B^{5/2}).
	\end{equation}
\end{lem}

\begin{proof}[Proof.]
	By straightforward calculations, it can be shown that $K(\theta,\mathbf{\tilde{\Sigma}}_f(t))$ can be split into:
	\begin{equation}
		K(\theta,\mathbf{\tilde{\Sigma}}_f(t)) = C_1(p,n) + F(\theta) + \sum C_2(p,n) G(\theta)
	\end{equation} 
	where $C_1(p,n)=\mathcal{O}_p(A)$ and $C_2(p,n)=\mathcal{O}_p(A^{1/2}B^{5/2})$ which are free of $\theta$, while $F(\cdot)$, $G(\cdot) \in \mathcal{C}^{\infty}$, according to Lemma 1. In particular, all the derivatives of $F$ and $G$ are bounded on $\Theta$. As a consequence, 
	\begin{eqnarray}
		\frac{\partial K(\theta,\mathbf{\tilde{\Sigma}}_f(t))}{\partial\theta_i} &=& \frac{\partial F(\theta)}{\partial \theta_i} + \sum C_2(p,n) \frac{\partial G(\theta)}{\partial \theta_i} \nonumber \\
		&=& \mathcal{O}_p(A^{1/2}B^{5/2}).
	\end{eqnarray}
\end{proof} 

\noindent
\begin{lem} The log-likelihood function for the CAW model given observed data $\mathbf{\hat{\Sigma}}_f(t)$ and true data $\mathbf{\tilde{\Sigma}}_f(t)$ are 
	\begin{eqnarray}
		\mathcal{\hat{L}}(\theta) &=& \frac{1}{T}\sum_{t=1}^{T}\{-\frac{\nu r}{2}\text{ln}(2) - \frac{r(r-1)}{4}\text{ln}(\pi)-\sum_{i=1}^{r}\text{ln}\Gamma(\frac{\nu+1-i}{2}) -\frac{\nu}{2} \text{ln}|\frac{\mathbf{\hat{S}}_f(t)}{\nu}| \nonumber \\
		&&  + (\frac{\nu-r-1}{2})\text{ln}|\mathbf{\hat{\Sigma}}_f(t)|-\frac{1}{2}\text{tr}(\nu\mathbf{\hat{S}}_f(t)^{-1}\mathbf{\hat{\Sigma}}_f(t))\}
	\end{eqnarray}
	and
	\begin{eqnarray}
		\mathcal{L}(\theta) &=& \frac{1}{T}\sum_{t=1}^{T}\{-\frac{\nu r}{2}\text{ln}(2) - \frac{r(r-1)}{4}\text{ln}(\pi)-\sum_{i=1}^{r}\text{ln}\Gamma(\frac{\nu+1-i}{2}) -\frac{\nu}{2} \text{ln}|\frac{\mathbf{S}_f(t)}{\nu}| \nonumber \\
		&&  + (\frac{\nu-r-1}{2})\text{ln}|\mathbf{\tilde{\Sigma}}_f(t)|-\frac{1}{2}\text{tr}(\nu\mathbf{S}_f(t)^{-1}\mathbf{\tilde{\Sigma}}_f(t))\},
	\end{eqnarray}
	respectively. Then we have
	\begin{equation}
		\mathcal{\hat{L}}(\theta) - \mathcal{L}(\theta) = \mathcal{O}_p(A^{1/2}B^{5/2})
	\end{equation}
\end{lem}

\begin{proof}[Proof.]
	By simple algebraic manipulations, we have
	\begin{eqnarray}
		\mathcal{\hat{L}}(\theta) = \mathcal{L}(\theta) &+& \frac{1}{T}\sum_{t=1}^{T} \{ -\frac{\nu}{2}(\text{ln}|\frac{\mathbf{\hat{S}}_f(t)}{\mathbf{S}_f(t)}|) + (\frac{\nu-r-1}{2})\text{ln}|\frac{\mathbf{\hat{\Sigma}}_f(t)}{\mathbf{\tilde{\Sigma}}_f(t)}| \nonumber \\
		&&  - \frac{1}{2}\text{tr}(\nu(\mathbf{\hat{S}}_f(t)^{-1}\mathbf{\hat{\Sigma}}_f(t)-\mathbf{S}_f(t)^{-1}\mathbf{\tilde{\Sigma}}_f(t)))\}
	\end{eqnarray}
	
	\noindent
	Since $\mathbf{S}_f(t)$ is the conditional mean of $\mathbf{\tilde{\Sigma}}_f(t)$, it can be proved that $max_{t}\parallel \mathbf{S}_f(t) \parallel_2 = \mathcal{O}_p(B)$ and $\mathbf{\hat{S}}_f(t) - \mathbf{S}_f(t) = \mathcal{O}_p(A^{1/2}B^{3/2})$ as $\mathbf{\tilde{\Sigma}}_f(t)$ and $\mathbf{\hat{\Sigma}}_f(t) - \mathbf{\tilde{\Sigma}}_f(t)$. The basic idea of the proof is the following. We can treat Equation (10) as a linear map from $(\mathbf{\tilde{\Sigma}}_f(t))$ to $(\mathbf{S}_f(t))$: $(\mathbf{S}_f(t))=\mathcal{L}\cdot(\mathbf{\tilde{\Sigma}}_f(t))$ where $\mathcal{L}$ is a linear operator with bounded norm. It follows that $\forall t$, $\parallel \mathbf{S}_f(t) \parallel_2 \le \alpha \sup_s \parallel \mathbf{\tilde{\Sigma}}_f(t) \parallel_2$ for some constant $\alpha$, which gives the required bound. A similar argument can be applied to $\mathbf{\hat{S}}_f(t) - \mathbf{S}_f(t)$.\\
	
	\noindent
	For the first term of the additional part, since $\mathbf{S}_f(t) = \mathcal{O}_p(B)$, we have
	\begin{equation}
		||\frac{\mathbf{\hat{S}}_f(t)-\mathbf{S}_f(t)}{\mathbf{S}_f(t)}||_2 = \mathcal{O}_p(A^{1/2}B^{1/2})
	\end{equation}
	Thus,
	\begin{equation}
		\text{ln}|\frac{\mathbf{\hat{S}}_f(t)}{\mathbf{S}_f(t)}| = (|\frac{\mathbf{\hat{S}}_f(t)}{\mathbf{S}_f(t)}|-1) - o_p(|\frac{\mathbf{\hat{S}}_f(t)}{\mathbf{S}_f(t)}|-1) = \mathcal{O}_p(A^{1/2}B^{1/2})
	\end{equation}
	
	\noindent
	For the second term, it is easy to prove that $\mathbf{\tilde{\Sigma}}_f(t) = \mathcal{O}_p(B)$, so that we have
	\begin{equation}
		||\frac{\mathbf{\hat{\Sigma}}_f(t)-\mathbf{\tilde{\Sigma}}_f(t)}{\mathbf{\tilde{\Sigma}}_f(t)}||_2 = \mathcal{O}_p(A^{1/2}B^{1/2})
	\end{equation}
	Thus,
	\begin{equation}
		\text{ln}|\frac{\mathbf{\hat{\Sigma}}_f(t)}{\mathbf{\tilde{\Sigma}}_f(t)}| = (|\frac{\mathbf{\hat{\Sigma}}_f(t)}{\mathbf{\tilde{\Sigma}}_f(t)}|-1) - o_p(|\frac{\mathbf{\hat{\Sigma}}_f(t)}{\mathbf{\tilde{\Sigma}}_f(t)}|-1) = \mathcal{O}_p(A^{1/2}B^{1/2}) 
	\end{equation}
	
	\noindent
	For the third term, first note that as $\mathbf{S}_f(t) \ge CC'$, there exists a constant $w>0$ such that the minimum eigenvalue of $\mathbf{S}_f(t)$ is larger than or equal to $w$. Consequently, $||\mathbf{S}_f(t)^{-1}||_2 \le \frac{1}{w}$ which is bounded. As a result
	\begin{eqnarray}
		&& ||\mathbf{\hat{S}}_f(t)^{-1}\mathbf{\hat{\Sigma}}_f(t)-\mathbf{S}_f(t)^{-1}\mathbf{\tilde{\Sigma}}_f(t)||_2 \nonumber \\
		&\le& ||\mathbf{\hat{S}}_f(t)^{-1}(\mathbf{\hat{\Sigma}}_f(t)-\mathbf{\tilde{\Sigma}}_f(t))||_2+||(\mathbf{\hat{S}}_f(t)^{-1}-\mathbf{S}_f(t)^{-1})\mathbf{\tilde{\Sigma}}_f(t)||_2 \nonumber \\
		&\le& ||\mathbf{\hat{S}}_f(t)^{-1}||_2||(\mathbf{\hat{\Sigma}}_f(t)-\mathbf{\tilde{\Sigma}}_f(t))||_2+||\mathbf{\hat{S}}_f(t)^{-1}||_2||\mathbf{\hat{S}}_f(t)-\mathbf{S}_f(t)||_2||\mathbf{S}_f(t)^{-1}||_2||\mathbf{\tilde{\Sigma}}_f(t)||_2 \nonumber \\
		&=& \frac{1}{w}\mathcal{O}_p(A^{1/2}B^{3/2}) + \frac{1}{w^2}\mathcal{O}_p(A^{1/2}B^{3/2})\mathcal{O}_p(B) \nonumber \\
		&=& \mathcal{O}_p(A^{1/2}B^{5/2}) 
	\end{eqnarray}
\end{proof}

\begin{proof}[Proof of Theorem 2.]
	From Lemma 3, we have $\mathcal{\hat{L}}(\theta)-\mathcal{L}(\theta) = K(\theta,\mathbf{\tilde{\Sigma}}_f(t)) = \mathcal{O}_p(A^{1/2}B^{5/2})$. By taking derivatives of both sides with respect to any parameters from $\mathbf{\theta}$, we have
	\begin{equation}
		\frac{\partial\mathcal{\hat{L}}(\theta)}{\partial\theta_i}=\frac{\partial\mathcal{L}(\theta)}{\partial\theta_i} +\frac{\partial K(\theta,\mathbf{\tilde{\Sigma}}_f(t))}{\partial\theta_i}
	\end{equation}
	By Lemma 2, we have
	\begin{eqnarray}
		0=\frac{\partial\mathcal{\hat{L}}(\mathbf{\hat{\theta}})}{\partial\theta_i} &=& \frac{\partial\mathcal{L}(\mathbf{\hat{\theta}})}{\partial\theta_i}+\mathcal{O}_p(A^{1/2}B^{5/2}) \nonumber \\
		&=&\frac{\partial\mathcal{L}(\mathbf{\tilde{\theta}})}{\partial\theta_i} + \sum_j \frac{\partial^2\mathcal{L}}{\partial\theta_i\partial\theta_j}(\mathbf{\tilde{\theta}})(\mathbf{\hat{\theta}}-\mathbf{\tilde{\theta}}) + o_p(\mathbf{\hat{\theta}}-\mathbf{\tilde{\theta}})+\mathcal{O}_p(A^{1/2}B^{5/2}) \nonumber \\
		&=&\sum_j \frac{\partial^2\mathcal{L}}{\partial\theta_i\partial\theta_j}(\mathbf{\tilde{\theta}})(\mathbf{\hat{\theta}}-\mathbf{\tilde{\theta}}) + o_p(\mathbf{\hat{\theta}}-\mathbf{\tilde{\theta}})+\mathcal{O}_p(A^{1/2}B^{5/2}) 
	\end{eqnarray}
	From Lemma 1, we know that $\partial^2\mathcal{L}(\mathbf{\tilde{\theta}})/\partial\theta_i\partial\theta_j$ are bounded for any $\mathbf{\tilde{\theta}}$. In addition, from Conditions (A4) and (A6), there exists some constant $\ell$ such that in probability,
	\begin{equation}
		|\frac{\partial^2\mathcal{L}}{\partial\theta_i\partial\theta_j}(\mathbf{\tilde{\theta}})|>\ell>0
	\end{equation}
	uniformly. As a result, we have
	\begin{equation}
		\mathbf{\hat{\theta}}-\mathbf{\tilde{\theta}} = \mathcal{O}_p(A^{1/2}B^{5/2})
	\end{equation}
\end{proof}

\newpage
\noindent
{\bf Figure Legends}:\\

\noindent
Figure 2:\\
The subfigure with title "11" shows the time series of the variance for the first factor, the subfigure with title "21" shows that of the covariance between the first and second factors, and so on.\\

\noindent
Figure 3:\\
The first subplot is for the first entry of the $6$-dimensional vector, and so on. For each of the 6 time series, we show the in-sample fit in the upper panel, where the solid line stands for the real data and the dashed line is the fitted series. The residuals are shown in the middle panel while the ACF and PACF of residuals are displayed in the lower panel. \\

\noindent
Figure 4:\\
Two plots for realized variances and two plots for realized covariances are shown in Figure $\ref{f4}$. AIG and BA are chosen for the realized variances while the realized covariances between AIG and C, and between AXP and BA are shown in the following figures.

\newpage
  \bibliography{mybibfile}
\end{document}